\documentclass[%
 reprint,
 amsmath,amssymb,
 aps,
 prx,
]{revtex4-2}

\usepackage{amsmath}
\usepackage[normalem]{ulem}
\usepackage[dvipsnames]{xcolor}
\usepackage{amsthm}
\usepackage{amssymb}
\usepackage{bm}
\usepackage{physics}

\usepackage{graphicx}
\usepackage{dcolumn}
\usepackage{bm}
\usepackage{hyperref}
\usepackage[capitalise]{cleveref}

\newtheorem{theorem}{Theorem}
\newtheorem{proposition}{Proposition}
\newtheorem{definition}{Definition}

\newcommand{\hilb}{\mathcal{H}}

\DeclareMathOperator*{\ExpectedValue}{\mathbb{E}}

\begin{document}

\title{Positive-definite parametrization of mixed quantum states with deep neural networks}

\author{Filippo Vicentini}
\email{filippo.vicentini@epfl.ch}
\author{Riccardo Rossi}%
\author{Giuseppe Carleo}%
\affiliation{
\'{E}cole Polytechnique F\'{e}d\'{e}rale de Lausanne (EPFL), Institute of Physics, CH-1015 Lausanne, Switzerland
}%

\date{\today}

\begin{abstract}
We introduce the Gram-Hadamard Density Operator (GHDO), a new deep neural-network architecture that can encode positive semi-definite density operators of exponential rank with polynomial resources.
We then show how to embed an autoregressive structure in the GHDO to allow direct sampling of the probability distribution. 
These properties are especially important when representing and variationally optimizing the mixed quantum state of a system interacting with an environment.
Finally, we benchmark this architecture by simulating the steady state of the dissipative transverse-field Ising model.
Estimating local observables and the R\'enyi entropy, we show significant improvements over previous state-of-the-art variational approaches.
\end{abstract}

\maketitle


\section{Introduction}
\label{sec:intro}

Neural Quantum States (NQS) leverage Neural Networks (NN) to encode the state of quantum systems~\cite{carleo2017Science}.
The success of those techniques relies on the Universal Approximation Theorems~\cite{Cybenko1989Representation,Leshno1993Repr,best2010simulating}, which theoretically suggests that a sufficiently large NN is able to efficiently approximate a broad class of functions.
In the last few years the field has seen important progress and developments among several axes. 
In particular, a lot of effort has been devoted to creating more powerful ansatzes: after the first works on shallow Restricted Boltzman Machines (RBM)~\cite{Glasser2018PRX,Lu19PRBTopologicalNQS,Nomura2021Symm}, deeper, more expressive networks have been implemented~\cite{Saito2018Bosonic}, as well as autoregressive convolutional~\cite{sharir2020deep,Wu2019PRL-AutoregClassical} and recurrent~\cite{Allah2020PRR-RNN} neural networks that allow for direct sampling of the Born amplitudes\footnote{Direct sampling is more efficient than Metropolis, as it avoids autocorrelation problems in the sampled chain. This is particularly important for Variational Monte Carlo techniques, which cannot rely on a large number of samples at each optimization step.}. 
To further improve the performance while lowering the computational cost, physical structure has been imposed on the ansatz by exploting translational~\cite{choo2018symmetries,roth2021group} and non-abelian~\cite{Vieijra2020PRL-NonAbelian,Vieijra2021PRB} symmetries.
Moreover, considerable effort has also been devoted to developing architectures that can encode the antisymmetry of a fermionic wavefunction~\cite{Choo2020Fermions,NomuraPRBFermionic,Stokes2020PRB,Pfau2020PRR,Spencer2020ArXiv,Hermann2020NatureChem,Yoshioka2021Fermionic,Inui2021PRR,Diluo2022PRL,Nys22Fermions}.

In parallel, higher-order optimization techniques for ground-state optimisation have been proposed~\cite{Martens2015Arxiv,Webber2021Hessian}.
NQS can also be used to investigate the dynamics, either through an explicit Time-Dependent Variational Principle (TDVP)~\cite{Carleo2012TDVP,Yuan2019TDVP}, which can be affected by dynamical instabilities~\cite{Hofmann2021TDVPNoise},  or through a recently proposed implicit scheme~\cite{Gutirrez2022Quantum-Implicit,Reh2021PRL}.
Lastly, applications to quantum state tomography have also been proposed~\cite{Torlai2018NatPhys,Torlai2019PRLQST,Melkani2020PRAEigenstate,Palmieri2020NPJ,Torlai2020PRR,AhmedPRL2021QST}.

This rapid pace of development has been focused around the field of closed, isolated systems, where the state to be approximated is a pure wavefunction. 
Contrarily, applications and developments to fields where the underlying state is a mixed density matrix, such as out-of-equilibrium setups, open quantum systems and mixed state tomography, have been rather limited so far.

In 2019, proposals to study the steady-state~\cite{vicentini2019prl} and dynamics~\cite{Hartmann2019PRLDissipative,Nagy2019PRL} of Markovian Open Quantum Systems with NQS were published.
These works parametrized the logarithm of the matrix elements $\log \rho(\sigma, \eta)$ in the computational basis $\sigma,\eta \in \{\uparrow, \downarrow\}^N$, so that the density matrix is given by
\begin{equation}
    \hat{\rho} = \sum_{\sigma, \eta} \exp[\log\rho(\sigma, \eta)]\ket{\sigma}\bra{\eta}.
\end{equation}
By parametrizing $\log\rho(\sigma,\eta)$ with a neural network, and by using a variational principle to approximate the dynamics or the steady-state, the computational cost can in general be made polynomial in the system size.
The most common approach relies on a shallow network, a purified RBM, which automatically satisfies the physical requirement of hermiticity and positive semi-definiteness~\cite{torlai2018latent}.
Another approach is to not enforce these physical requirements~\cite{Yoshioka2019PRB}.
These pioneering works suggest that enforcing physical requirements significantly improves the performance.

While the purified RBM approach of Ref.~\cite{torlai2018latent} possesses many desirable properties, 
it is, however, ultimally limited as it can only be applied to shallow networks with one layer.
Increasing the depth is often necessary to build more advanced architectures such as autoregressive or recurrent networks, and is widely expected to improve the expressive power of the ansatz~\cite{Levine2019}.
An architecture based on a purified Deep Boltzmann Machine has been proposed by Nomura and coworkers in order to represent Gibbs States~\cite{Nomura2021PRL}. In the general case, however, this approach scales exponentially with the number of deep hidden units. 

The most advanced, state-of-the-art results~\cite{Reh2021PRL,Luo2022PRL} have been obtained with the Positive-Operator Valued Measurement (POVM) ansatz of Carrasquilla and coworkers~\cite{Carrasquilla2019POVM}, where the density matrix is expanded in the basis $\{ a \}$  of the outcomes of a set of measurement operators $\{ \hat{M}_a \}$, given by

\begin{equation}
    \hat{\rho}= \sum_a p(a) \hat{M}_a.
\end{equation}
The advantage of this representation is that one only needs to variationally approximate the positive probability distribution $p(a)$, which can also be parametrized with autoregressive, recurrent and deep networks.
While it is possible to enforce local positivity by carefully chosing $\hat{M}_a$, however, the density matrix so obtained is not in general, a positive-definite and valid physical state.

Lastly, we mention a different approach originally discussed by Melkani et al. in Ref.~\cite{Melkani2020PRAEigenstate} and more recently by Donatella et al. in Ref.~\cite{Donatella2022Private,Donatella2021PRA}, where the density matrix is expanded onto a finite number of pure-states which are approximated with an NQS.
As the entropy of a generic mixed quantum state scales linearly with the system size, its rank scales exponentially, requiring thus an untractable computational effort for large systems.

At the time of writing and to the extent of our knowledge, a flexible ansatz that can represent arbitrary-rank, positive semi-definite density matrices, whose depth can be tuned, does not exist.

In this work we introduce the class of Gram-Hadamrd density operators, based on deep neural networks, which ensure the positive semi-definiteness of the density matrix.
We also show that by embedding an autoregressive structure it is possible to perform direct sampling of the diagonal probability distribution.
We then present numerical results for the dissipative transverse-field Ising model to benchmark our architecture and compare it with previous approaches.
Estimating local observables, we find that the reduced 1-spin density matrix is well reproduced by this ansatz, and by computing the Renyi-2 entropy we assess that we are able to capture some global properties as well.

This article is structured as follows: after this introduction, we briefly recall some mathematical concepts that are useful to impose  positive-definitness (\cref{sec:algebra-prelim}); we introduce the Gram-Hadamard architectures in \cref{sec:gram-hadamard-ansatz} and \cref{sec:autoreg-ansatz}; in \cref{sec:tdvp} we discuss the variational time-evolution of these NDOs; in \cref{sec:results-ising} we present the results of the numerical experiments; finally, in \cref{sec:conclusion} we present the conclusions and outlooks of this work.

\section{Deep ansatzes for mixed quantum states}
\label{sec:deep-ansatz}

\subsection{Algebraic Preliminaries}
\label{sec:algebra-prelim}
We start by briefly recalling some results of linear algebra.
We start with the definition of a Gram matrix:

\begin{definition}[Gram matrix]
The \textbf{Gram matrix} of a sequence of $R$ vectors $(\bm{\psi}_a)_{a\in\{1,\dots ,R\}}$ of dimension $n$ is
the $n\times n$ matrix $G$ given by
\begin{equation}
    G = A\, A^\dagger,
\end{equation}
where $A_{i,a}=({\psi}_a)_i$, $j\in\{1,\dots,n\}$. 
Notably, $G$ is positive semi-definite and $\text{Rank}[G] \leq R$. 
\end{definition}
We remark that any hermitian positive semi-definite matrix can be written as a Gram matrix.
We then give a formal definition of the element-wise product:

\begin{definition}[Hadamard product]
We define the Hadamard (or element-wise) product $H=A\odot B$ of two matrices $A$ and $B$ to be:
\begin{equation}
H_{i,j} = (A\odot B)_{i,j} := A_{i,j}\, B_{i,j}.
\end{equation}
\end{definition}
Interestingly, a well-known result in linear-algebra is that the Hadamard product of matrices preserves the positive definiteness

\begin{theorem}[Shur’s product theorem]
Let $A$ and $B$ be hermitian positive semi-definite matrices. Then, $A\odot B$ is also hermitian positive semi-definite.
\end{theorem}
The proof is presented, for completeness, in appendix~\ref{sec:appendix-shur-product}.

In the discussion above, we have presented two approaches to construct positive semi-definite matrices: by (i) building a Gram matrix from a sequence of vectors or by (ii) using the Hadamard product. 
The Gram-Hadamard matrix is obtained by combining the two approaches in order to create a positive semi-definite matrix:

\begin{definition}[Gram-Hadamard matrix]\label{definition-gh}
The \textbf{Gram-Hadamard matrix} $S$ of a set of sequences of vectors $\{(\bm\psi_a^{(h)})_{a=1,\dots,R_h}\}_{h=1,\dots,K}$ is obtained by taking the Hadamard product of the Gram matrices of the sequences of vectors in the set:,
\begin{equation}
    S = \bigodot_{h=1}^K A^{(h)}
    \,[A^{(h)}]^\dagger,
\end{equation}
where $[A^{(h)}]_{i,a}=[\psi_{a}^{(h)}]_i$. More explicitely
\begin{equation}
    S_{i,j} = \prod_{h=1}^K \sum_{a=1}^{R_h}\psi_{i,a}^{(h)}\,
    [\psi_{j,a}^{(h)}]^*.
\end{equation}
where we have defined $\psi_{i,a}^{(h)}:=[\psi_a^{(h)}]_i$.
\end{definition}
The hermiticity and positive-definiteness of the \textbf{Gram-Hadamard matrix} follows from the Shur's product theorem. The maximum rank of $S$ is bounded exponentially in $K$, as $\text{Rank}[S] \leq \prod_{h=1}^K R_h$, therefore allowing to represent a positive semi-definite matrix of  exponential rank with a polynomial number of parameters.

For completeness, we also point out that is possible to further transform the output of the GHDO by element-wise application of a nonlinear function without losing the positive semi-definite property, provided that such function satisfies the condition described in \cref{sec:appendix-activation-function}. 


\subsection{Gram-Hadamard Density Operator}
\label{sec:gram-hadamard-ansatz}

\begin{figure}
    \centering
    \includegraphics{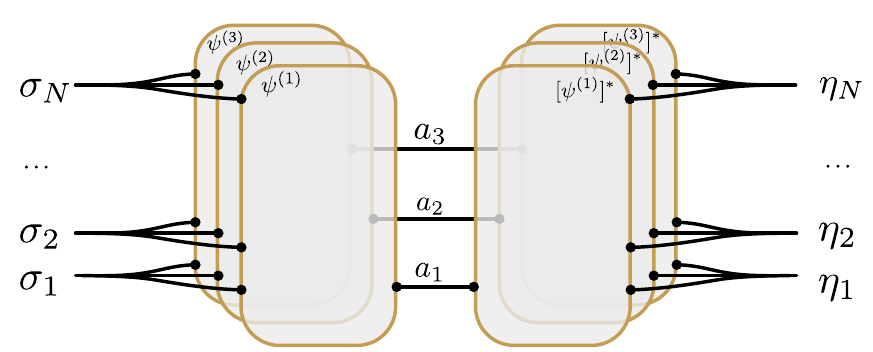}
    \caption{Schematic drawing of the Gram-Hadamard Density Operator defined in \cref{eq:def:gram-hadamard-ndo} using Tensor-Network notation. The drawing emphasizes that this ansatz is equivalent to the product of several plaquettes as large as the system, but with a limited rank $R=\dim[a_i]$.}
    \label{fig:schema_plaquette}
\end{figure}
\subsubsection{Motivation}
A density operator must satisfy the requirements of hermiticity and positive semi-definiteness in order to describe a physical state.
For this reason, we would like to devise a variational ansatz that automatically respects these two conditions.

A common way to achieve hermiticity and positive semi-definiteness is by using the so-called \textit{purification ansatz}, which constructs a positive semi-definite density operator  $\hat\rho:\hilb\to\hilb$ (acting on the the physical Hilbert space $\hilb$) by embedding $\hat{\rho}$ in an extended Hilbert space where the quantum state is pure, $|\psi\rangle\in\hilb \otimes\hilb_a$. 
Then, the ancillary space $\hilb_a$ is traced away in order to generate the mixedness,
\begin{equation}
    \bra{\sigma}\hat\rho\ket{\eta} = \sum_{a=1}^{\text{dim}[\hilb_a]} \langle \sigma,a|\psi\rangle\,\langle\psi|\eta,a\rangle,
\end{equation}
where $\{|\sigma\rangle\}$ ($\{|a\rangle\}$) is a basis of $\hilb$ ($\hilb_a$).
If we introduce the matrix $A_{\sigma,a}=\langle \sigma,a|\psi\rangle$, then $\hat\rho=A \,A^\dagger$ is a Gram matrix and its maximum rank is bounded by $\dim[\hilb_a]$.

A good variational ansatz for the density matrix must satisfy the \textit{query-access} property, meaning that matrix elements $ \bra{\sigma}\hat\rho\ket{\eta} $ must be computable in polynomial time. The entropy of a generic mixed quantum state will be linear in system size, which means that  $\text{dim}[\mathcal{H}_a]$ must be exponential in system size. 
This requires that the trace over the ancilla be performed either analitically, which is possible in some special case, or by randomly sampling the ancilla state. 
However, sampling the acilla is efficient only if the wave function is non-negative, an assumption that is generically not verified, especially for time-evolved states.


\subsubsection{Introducing the GHDO}
\begin{definition}[GHDO] Following Definition~\ref{definition-gh}, we introduce the \textbf{Gram-Hadamard Density Operator} (GHDO) as
\begin{equation}
    \label{eq:def:gram-hadamard-ndo}
    \bra{\sigma}\hat\rho\ket{\eta} = \prod_{h=1}^K \sum_{a=1}^{R}\psi_{\sigma,a}^{(h)}\, [\psi_{\eta,a}^{(h)}]^*,
\end{equation}
where $\psi_{\sigma,a}^{(h)}\in\mathbb{C}$. 
\end{definition}

The GHDO is a hermitian,  positive semi-definite density matrix, whose elements are calculable in polynomial time in $K$ and $R$, and whose rank is bounded exponentially by $R^K$.
We remark that it is possible to exactly represent the maximally mixed state of a system of $N$ spins, $\hat{\rho}=\frac{\mathbb{I}}{2^N}$, with an GHDO with $R\geq2$ and $K\geq N$.

$\psi_{\sigma, a}^{(h)}$ can then be parametrized with an unconstrained and arbitrarily-deep neural network. 
It is possible to show that this form is more general than the NDOs form of Ref.~\cite{torlai2018latent}, which can be reproduced by a GHDO with $R=2$ and $\psi_{\sigma, a=1}^{(h)} = 1$.

\subsection{Autoregressive Gram-Hadamard Density Operator}
\label{sec:autoreg-ansatz}
\subsubsection{Motivation}
\begin{figure}
    \centering
    \includegraphics{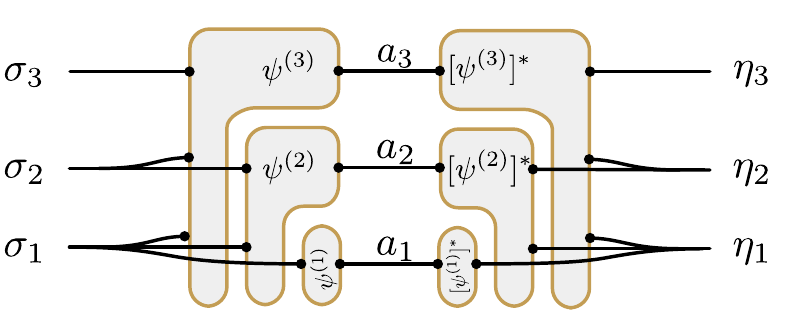}
    \caption{Schematic drawing of the Autoregressive Gram-Hadamard Density Operator as defined by \cref{eq:definition-aghdo} using Tensor-Network notation.
    The diagram emphasizes the autoregressive order.}
    \label{fig:schema_autoregressive}
\end{figure}
The density operator is used to calculate expectation values of operators.
An efficient way of doing so for a {$k$-local} operator $\hat{A}$ is by obtaining a sample $\sigma$ in the computational basis, and analytically calculating the non-diagonal contribution of $\hat{A}$
\begin{equation}
    \langle\hat{A}\rangle := \frac{\text{Tr}\;\hat{\rho}\,\hat{A}}{\text{Tr}\;\hat{\rho}} = 
    \frac{\sum_{\sigma}\langle \sigma | \hat{\rho}|\sigma\rangle 
    \;A_{\text{loc}}(\sigma)
    }{\sum_\sigma \langle\sigma|\hat{\rho}|\sigma\rangle} = \ExpectedValue_{\sigma\sim p(\sigma)}\left[ A_{\text{loc}}(\sigma)\right],
\end{equation}
where
\begin{equation}
A_{\text{loc}}(\sigma):=\sum_{\eta}\frac{\langle \sigma|\hat{A}|\eta\rangle\langle\eta|\hat{\rho}|\sigma\rangle}{\langle\sigma|\hat{\rho}|\sigma\rangle},
\end{equation}
\begin{equation}
    p(\sigma) := \frac{\langle\sigma|\hat{\rho}|\sigma\rangle}{\sum_{\eta}\langle\eta|\hat{\rho}|\eta\rangle},
\end{equation}
and where $\ExpectedValue$ denotes the classical expected value.
For this task, we need to draw samples from the probability distribution $p(\sigma)$ defined by the diagonal of the density matrix.

The GHDO we have introduced in \cref{sec:gram-hadamard-ansatz} allows to draw samples from $p(\sigma)$ with the Markov-chain Monte Carlo method. However, such sampling techniques are prone to thermalization and autocorrelation problems when used in variational contexts, as $\hat{\rho}$ changes at each optimization step. 
To avoid such problems one traditionally resorts to drawing a  very large number of samples for a single step, which can be computationally expensive.
For these reasons, it would be convenient to restrict the GHDO to a form that allows direct sampling, that is, drawing completely uncorrelated samples from $p(\sigma)$.
A way of achieving direct sampling is rewriting the probability distribution $p(\sigma)$ over $N$ spins in the autoregressive form, i.e. as a product of conditional probabilities
\begin{equation}
    p(\sigma) = \prod_{h=1}^{N} p(\sigma_h|\sigma_{<h})
\end{equation}
where $\sigma=(\sigma_1,\dots,\sigma_{N})$ and $\sigma_{<i}=(\sigma_1,\dots,\sigma_{i-1})$. It is possible to directly sample a probability distribution in the autoregressive form with the 
the following technique: we first sample $\sigma_1$ from $p(\sigma_1|\emptyset)$, we then use $\sigma_1$ to sample $\sigma_2$ from $p(\sigma_2|\sigma_1)$, and so on. Autoregressive neural-network wave functions based on the autoregressive property were originally proposed by Sharir and coworkers~\cite{sharir2020deep}.

\subsubsection{Introducing the AGHDO}
\label{sec:autoreg-ghdo}
The previous section motivates the introduction of the {\bf Autoregressive Gram-Hadamard Density Operator} (AGHDO) ansatz for a system of $N$ spins:
\begin{definition}[AGHDO]
\begin{equation}\label{eq:definition-aghdo}
    \langle \sigma|\hat{\rho}|\eta\rangle = \prod_{h=1}^N\sum_{a=1}^R \psi_{\sigma_{\le h},a}\,\left[\psi_{\eta_{\le h},a}\right]^*
\end{equation}
where ${\sigma=(\sigma_1,\dots,\sigma_N)}$ is the spin configuration, ${\sigma_{\le h}=(\sigma_0,\dots,\sigma_h)}$, and, for each $1\le h\le N$, we impose
\begin{equation}\label{eq:constraint-aghdo}
\sum_{\sigma_h\in\{-1,1\}}    \sum_{a=1}^R|\psi_{\sigma_{\le h},a}|^2 = 1
\end{equation}
for all values of $\sigma_{<h}$. With this choice, we have $\text{Tr}\;\hat{\rho}=1$ and 
\begin{equation}
    \langle \sigma|\hat{\rho}|\sigma\rangle = \prod_{h=1}^{N}p(\sigma_h|\sigma_{<h}),
\end{equation}
where
\begin{equation}
    p(\sigma_h|\sigma_{<h})=\sum_{a=1}^R |\psi_{\sigma_{\le h},a}|^2.
\end{equation}
We refer to $R$ as the local rank of the AGHDO.
\end{definition}



To impose the constraint of \cref{eq:constraint-aghdo} it is sufficient to use an unconstrained $\phi_{\sigma_{\le h},a}$ and then normalize it with a computational cost linear in the local Hilbert-space dimension:
\begin{equation}
    \psi_{\sigma_{\le h},a}=\frac{\phi_{\sigma_{\le h},a}}{\sqrt{\sum_{\sigma_h}\sum_a|\phi_{\sigma_{\le h},a}|^2}}.
\end{equation}
In the numerical implementation, the function $\phi_{\sigma_{\le h},a}$ is represented by an arbitrary neural network.

\subsubsection{Properties}
The main hyperparameter of the AGHDO ansatz introduced in \cref{eq:definition-aghdo} is the ``local rank'' $R$ of the density operator.
In order to perform controlled numerical calculations, it is possible to verify that as $R$ gets larger, one is able to represent all relevant quantum states, as proven in \cref{sec:appendix-universal}.

We remark that when the local rank $R$ is equal to one, the rank of the AGHDO is one, which means that the density operator is a pure state.
It is easy to verify that the quantum state one obtains in this case is equivalent to the autoregressive neural wave function introduced in~\cite{sharir2020deep}.
From the known properties of the wave function in the autoregressive form, this correspondence shows that the AGHDO with $R=1$ can represent an arbitrary pure state.

Another limiting interesting case are classical states $\hat{\rho}_{\text{cl}}$, which are diagonal in the computational basis,
\begin{equation}
    \label{eq:definition-classical}
    \langle \sigma|\hat{\rho}_{\text{cl}}|\eta\rangle = p(\sigma)\;\delta_{\sigma,\eta}.
\end{equation}
A classical state of $N$ spins generically has a rank exponential in $N$.
We remark that it can be exactly represented by an AGHDO with a constant local rank $R=2$, as discussed in \cref{sec:appendix-classical}.

\section{Variational Time Evolution}
\label{sec:tdvp}

\subsection{Time evolution of mixed quantum states}
In this work we focus on the challenging problem of the numerical simulation of the dynamics of a mixed many-body quantum state interacting with a Markovian bath, which is described by an ordinary differential equation for the density operator $\hat{\rho}$
\begin{equation}\label{eq:time-evolution-superoperator}
    \frac{d\hat{\rho}}{dt} = \mathcal{L}\hat{\rho},
\end{equation}
where $\mathcal{L}$, the Liouvillian, is a {\it super-operator}, i.e. a linear operator acting on operators, defined by its action on $\hat{\rho}$, 
\begin{equation}
    \label{eq:lindblad}
    \mathcal{L}\hat{\rho}= -i\comm{\hat{H}}{\hat{\rho}} -\frac{1}{2}\sum_i\acomm{\hat{L}_i^\dagger \hat{L}_i}{\hat{\rho}} + \sum_i\hat{L}_i\hat{\rho}\hat{L}_i^\dagger,
\end{equation}
where $\hat{H}$ is the hamiltonian and $\hat{L}_i$ are the jump operators.
The discussion of this section also applies to unitary evolution, which can be obtained in this formalism by setting the jump operators $L_i$ to zero, and to  imaginary-time evolution, which requires the superoperator $\mathcal{L}_{\text{imag}}\hat{\rho}=-\{\hat{H},\hat{\rho}\}$.

An exact numerical simulation of \cref{eq:time-evolution-superoperator} for $N$ spins would require time and memory computational resources scaling like $2^{2N}$, and it is therefore intractable even for a moderate number of spins.
When working with variational density operators $\hat{\rho}_{\bm{w}}$, parametrized by a vector $\bm w\in\mathbb{R}^d$, the McLachlan variational principle can be used to recast \cref{eq:time-evolution-superoperator}, which has $2^{2N}$ terms, to a $d$-dimensional differential equation for the variational parameters $\bm w$~\footnote{This equation is valid for both real and complex parameters. See~\cite{Yuan2019TDVP}, Table 1 and following discussions.} 
\begin{equation}\label{eq:variational-time-evolution}
    \frac{d\bm{w}}{dt} = S^{-1}\bm{F},
\end{equation}
where the so-called quantum geometric tensor $S$ and the vector $\bm F$ can be stochastically estimated as follows
\begin{align}
    \label{eq:S-matrix-def}
    S_{i,j} &= \ExpectedValue_{\rho^2}[O_{i}^*(\sigma, \eta)O_j(\sigma, \eta)]-
    \ExpectedValue_{\rho^2}[ O_{i}^*(\sigma, \eta)]\ExpectedValue_{\rho^2}[O_j(\sigma, \eta)], \\
    \label{eq:F-vector-def}
    F_{i} &= \ExpectedValue_{\rho^2}[ O_{i}^*(\sigma, \eta)\mathcal{L}_\text{loc}(\sigma, \eta)] - \ExpectedValue_{\rho^2}[ O_{i}^*(\sigma, \eta)]\ExpectedValue_{\rho^2}[\mathcal{L}_\text{loc}(\sigma, \eta)],
\end{align}
where $\ExpectedValue_{\rho^2}[\cdot] = \ExpectedValue_{\sigma,\eta\sim\abs{\rho}^2(\sigma, \eta)}[\cdot] $ is the classical average over the $\abs{\rho}^2(\sigma, \eta)$ distribution
\begin{equation}\label{eq:average-superoperator}
    \ExpectedValue_{\rho^2}[f(\sigma,\eta)] := \frac{\sum_{\sigma,\eta}|\langle\sigma|\hat{\rho}|\eta\rangle|^2\,f(\sigma,\eta)}{\sum_{\sigma,\eta}|\langle\sigma|\hat{\rho}|\eta\rangle|^2}
\end{equation}
$O_{i}(\sigma, \eta)=\partial_{w_i}\log[\rho(\sigma, \eta)]$ is the logarithmic derivative with respect to the parameters, and the local estimator of the Liouvillian $\mathcal{L}_\text{loc}$ has the form
\begin{equation}
    \label{eq:L-loc}
    \mathcal{L}_{\text{loc}}(\sigma,\eta) := \frac{\bra{\sigma}\mathcal{L}\hat{\rho}\ket{\eta}}{\bra{\sigma}\hat{\rho}\ket{\eta}}.
\end{equation}

In this work, we obtain the steady-state density operator by performing the long-time simulation of \cref{eq:variational-time-evolution}, instead of using a variational principle for the steady state as was done, e.g., in Ref.~\cite{vicentini2019prl}.
To reduce the computational cost we heavily regularise \cref{eq:variational-time-evolution} in order to obtain weak-convergence~\footnote{Weak convergence, in the sense of Stochastic Differential Equations, means that we converge to the correct steady-state but the solution at intermediate times does not necessarily represents the physical dynamics.} to the steady-state as originally proposed in Ref.~\cite{Nagy2019PRL}.

\subsection{Sampling superoperator averages}
We have shown in \cref{sec:autoreg-ghdo} that it is possible to directly sample configurations $\sigma\sim\rho(\sigma, \sigma)$ from the diagonal of the AGHDO.
However, to perform the variational time evolution, samples from $\abs{\bra{\sigma}\hat{\rho}\ket{\eta}}^2$ are needed in order to estimate $S$ and $\bm{F}$ according to \cref{eq:S-matrix-def,eq:F-vector-def}.

We first remark that for a Trace-1 density operator such as the AGHDO, the normalization of the superoperator probability distribution, $\sum_{\sigma,\eta}|\langle\sigma|\hat{\rho}|\eta\rangle|^2=\text{Tr}\,\hat{\rho}^2$, is the purity of the density operator, and it is generically different from one.
However, this is irrelevant for  \cref{eq:variational-time-evolution} as the solution is independent of such normalization constant. 

We remark that, for a pure-state, the superoperator probability distribution factors as $|\langle\sigma|\hat{\rho}|\eta\rangle|^2 = p(\sigma)\,p(\eta)$
, while for a classical state we have $|\langle\sigma|\hat{\rho}|\eta\rangle|^2 = p(\sigma)^2\,\delta_{\sigma,\eta}$.
This motivates the use of a probability distribution that is intermediate between these two distributions, and a natural choice is the convex combination of \textit{conditional} probabilities, which depends on the convexity parameter $\alpha\in[0,1]$
\begin{equation}
    p_{\alpha}(\sigma,\eta) :=p(\sigma) \prod_{h=1}^N\left[\alpha \,p(\eta_h\,|\,\eta_{<h})+(1-\alpha)\,\delta_{\sigma_h,\eta_h}\right].
\end{equation}
We then use the distribution $p$ to perform importance sampling of the superoperator averages
\begin{equation}
    \label{eq:sampling-rho2-hybridization}
    \ExpectedValue_{\rho^2}[ f(\sigma,\eta)]\propto\ExpectedValue_{\sigma,\eta\sim p_\alpha(\sigma,\eta)}\left[\frac{|\langle\sigma|\hat{\rho}|\eta\rangle|^2}{p_\alpha(\sigma,\eta)}\;f(\sigma,\eta)\right].
\end{equation}
The probability distribution $p_\alpha$ allows direct sampling (see appendix~\ref{sec:appendix-hybridization}). When $p_\alpha(\sigma,\eta)$ is zero, $|\langle\sigma|\hat{\rho}|\eta\rangle|^2$
is also zero (more precisely $|\langle\sigma|\hat{\rho}|\eta\rangle|^2/p_{\alpha}(\sigma,\eta)\le \alpha^{-N}$), and the variance of the estimator in \cref{eq:sampling-rho2-hybridization} is bounded.

As the variance of the estimator depends on the purity of the state, one can improve the statistical estimate of observables by taking $\alpha\rightarrow 0$ when the purity is small, and $\alpha\rightarrow 1$ when the purity approaches 1.

\section{Numerical results}
\label{sec:results-ising}

To prove the viability of the AGHDO ansatz, we consider the dissipative transverse-field Ising model, defined by the hamiltonian
\begin{figure}
    \centering
    \includegraphics[width=\linewidth]{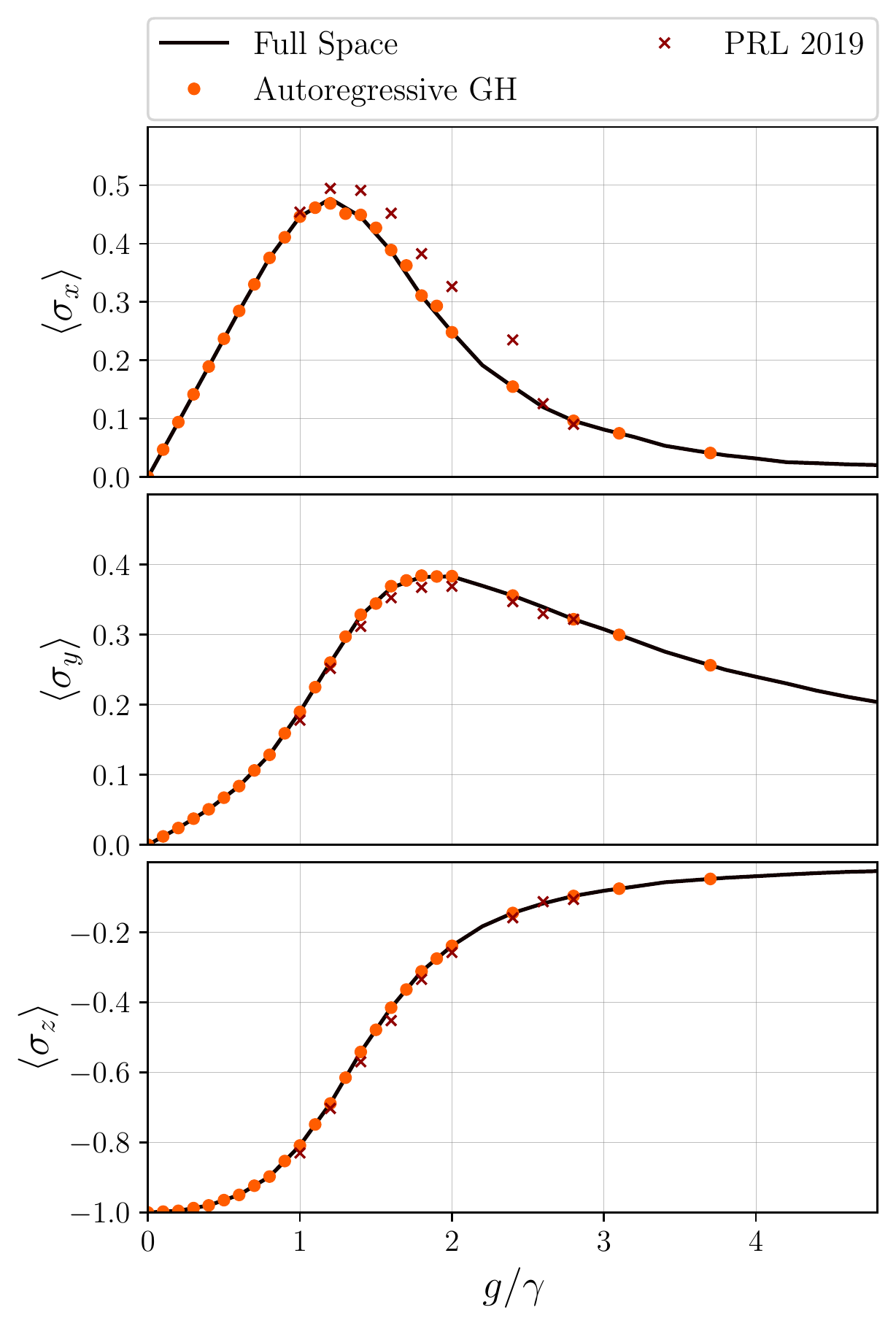}
    \caption{Average magnetization along the 3 axes $\hat{x}$, $\hat{y}$ and $\hat{z}$ in a 1D periodic chain with 16 sites, computed with the Autoregressive Gram-Hadamard Density Operator and compared with the results obtained using the Purified NDO in \cite{vicentini2019prl}. For the points in $g\in[1.0, 2.5]$ we used a 3 layers Networks with feature densities $[8,4,4]$ and local rank $R=32$. For the other points we used $R=8$. We remark that we considerably improve on the previous results in the region where the previous architecture was failing. Other parameters are $V/\gamma=2$.}
    \label{fig:ising_magnetization}
\end{figure}

\begin{equation}
    \hat{H} = \frac{V}{4}\sum_{\langle i,j\rangle} \hat{\sigma}^{z}_i\hat{\sigma}^{z}_j + \frac{g}{2}\sum_i \hat{\sigma}^{x}_i,
\end{equation}
and the local jump operators $\hat{L}_i=\sqrt{\gamma}\,\hat{\sigma}_i^{-}$ acting on each site according to the Lindblad Master Equation \cref{eq:lindblad}.
The model is believed to exhibit a first-order dissipative phase transition in $D\geq2$ when scanning the transverse field $g$~\cite{Jin2018PRBIsing,Daniel2021PRAIsingField}. 

Due to the numerical challenge that this model provides, the one-dimensional periodic chain with $N=16$ has already been used as a benchmark for NDOs in the past~\cite{vicentini2019prl,Luo2022PRL}.
We pick the same parameters of $V/\gamma=2$ used in Ref.~\cite{vicentini2019prl} and we simulate a heavily-regularized dynamics until convergence to the steady-state is achieved.
We expect the points in the interval  $g/\gamma\in[1.0, 2.5]$ to be the hardest to simulate, as the gap of the Liouvillian is smaller in this region \cite{Minganti2018Spectral}.

\begin{figure}
    \centering
    \includegraphics[width=\linewidth]{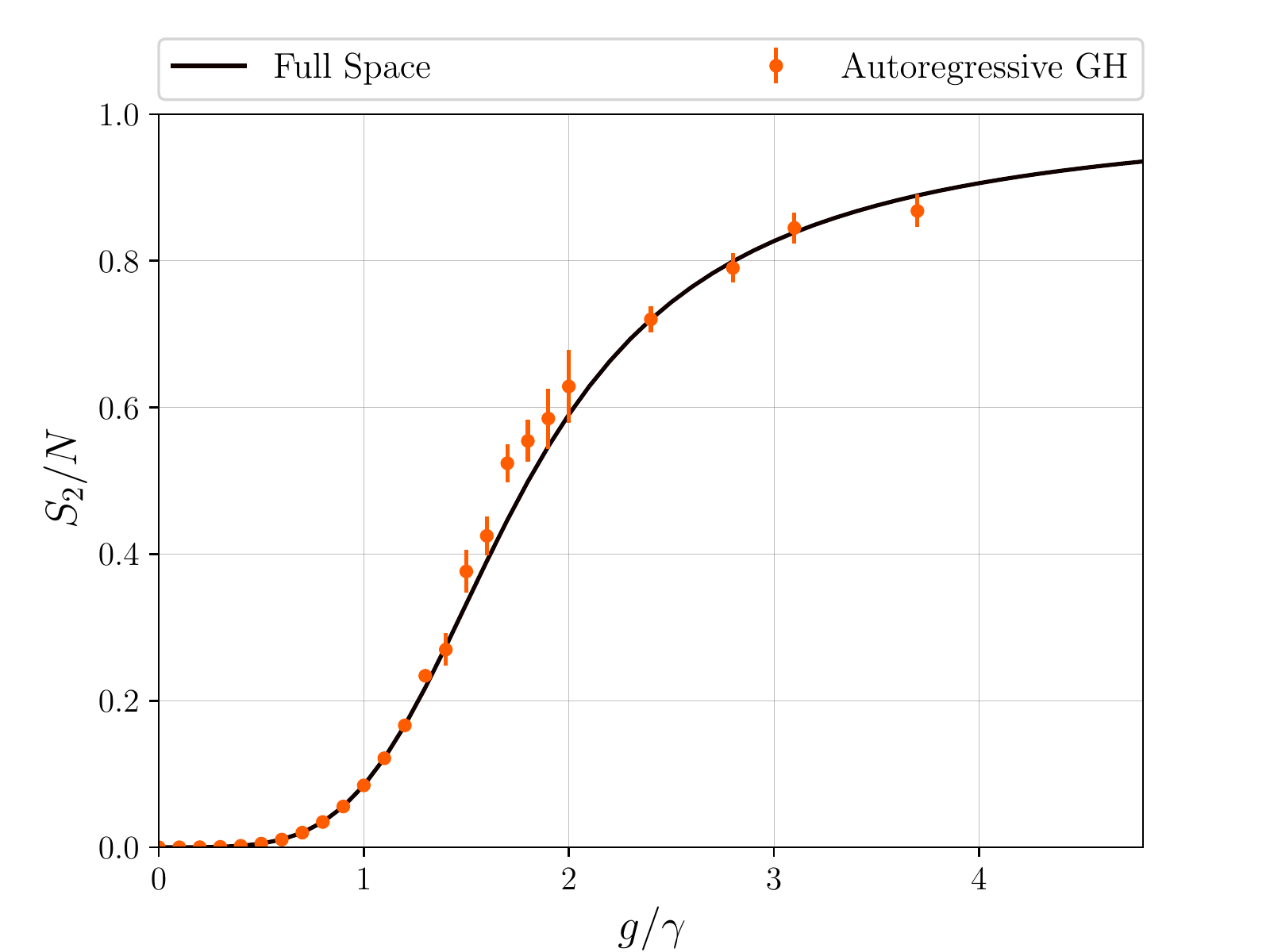}
    \caption{R\'enyi-2 entropy $S_2 = - \log_2[\Tr[\rho^\dagger\rho]]$ computed in the steady-state of a $1D$  periodic chain with $16$ sites of the Transverse-Field Ising Model.
    The error bar represents the statistical error in the estimation of the observable using $2^{15}$ samples. Parameters are the same as in \cref{fig:ising_magnetization}.}
    \label{fig:ising_reny2}
\end{figure}

For points outside of the interval $g/\gamma\in[1.0, 2.5]$ we used an Autoregressive Gram-Hadamard Density Operator with 2 masked dense layers, feature density $[8,4]$, and local rank $R=8$.
For points inside the interval we considered an extra layer, obtaining a feature density of $[8,4,4]$, and increased the local rank to $R=32$.
We sample the probability distribution $p_{\alpha}$ with $\alpha\in\{0.2, 0.5, 0.8\}$ depending on the purity of the state. The total number of samples for the time-evolution varies between $2^{12}$ and $2^{15}$, and the timestep is $10^{-3}$.

To assess the accuracy of the approximation of the reduced 1-spin density matrix, we present in \cref{fig:ising_magnetization} a comparison of the average magnetization along the three axes between the AGHDO and the results of Ref.~\cite{vicentini2019prl}.
The data shows a clear improvement over the shallow ansatz used in that reference.
We also investigated the performance in the approximation of global properties of the density matrix, by estimating the base-2 R\'enyi-2 entropy
\begin{equation}
    S_2 = -\log_2[\Tr[\hat{\rho}^\dagger\hat{\rho}]],
\end{equation}
which is proportional to the impurity of the state (see \cref{fig:ising_reny2}).
Comparing against full-space simulations we show that the AGHDO can capture the mixed-nature of the state to very high accuracy. 
As expected, the importance-sampling distribution we use, $p_\alpha$, works best for high- and low- purity states, while it has higher variance in the intermediate region.

\section{Conclusions}
\label{sec:conclusion}

In this work we have introduced a family of  neural-network architectures, the Gram-Hadamard density operators, which can be made arbitrarily deep while preserving the properties of physical density matrices, namely positive semi-definiteness and rank exponential in the system size.
We have also shown that it is possible to add an autoregressive structure to the GHDO that allows direct sampling of the diagonal of the density matrix.
Finally, we have presented numerical evidence that such construction is more expressive than previous variational architectures by considering the case of the dissipative transverse-field Ising model, where we have shown that we can accurately compute local observables and satisfactorily reproduce a global property, the purity of the density operator.

Future developments could investigate the integration of recurrent cells in the AGHDO, as well as the addition of the nonlinear functions that we briefly discussed in \cref{sec:appendix-activation-function}.

\section{Aknowledgements}
F.V. thanks Borgo Eibn for the original inspiration of this work and Zakari Denis for relevant discussions.

\appendix

\section{Proof of Shur's Product theorem}
\label{sec:appendix-shur-product}
{\bf Theorem:}
If $A$ and $B$ are positive semi-definite,
$A\odot B$ is positive semi-definite.
\begin{proof}
By using the Gram decomposition of the $n\times n$ matrices $A$ and $B$, we can write
\begin{equation}
    (A\odot B)_{ij}
    =(CC^\dagger\odot D D^\dagger)
    _{ij} =  \sum_{k=0}^{n^2-1}E_{ik}\,E_{jk}^*,
\end{equation}
where $E_{ik}:=C_{i,\lfloor k/n\rfloor+1}\;D_{i,(k\,\text{mod}\,n)+1}$.
\end{proof}

\section{Positive semi-definite matrices by application of a ``positive nonlinear function''}\label{sec:appendix-activation-function}
For completeness, we briefly discuss an interesting additional way to create a positive semi-definite matrix by element-wise application of a special type of nonlinear functions. This could be used to possibly further reduce the local rank of the AGHDO.
\begin{proposition}
Let $f:\mathbb{C}\to\mathbb{C}$ be an analytic function in a disk of radius $R$ around the origin. We can then write $f(z)=\sum_{l=0}^{\infty}a_l\,z^l$ for $|z|<R$. Suppose further that $a_l\ge 0$ for all $l$. Let $A$ be a positive semi-definite hermitian matrix such that $\max_{i,j}|A_{i,j}|<R$. Let $B$ be a matrix of the same dimensions of $A$ defined by
\begin{equation}
    B_{i,j}=f(A_{i,j})
\end{equation}
Then, $B$ is positive semi-definite and hermitian.
\end{proposition}
\begin{proof}
We consider the truncated $B^{(k)}$ matrix defined by
\begin{equation}
    B_{i,j}^{(k)} = \sum_{l=0}^k a_l\,A_{i,j}^l
\end{equation}
$B^{(k)}$ is positive semi-definite and hermitian as the sum of hadamard products of positive semi-definite hermitian matrices. Let $A_M:=\max_{i,j}|A_{i,j}|$, $A_M<R$. For $\epsilon>0$, let $k_\epsilon$ such that ${|f(z)-\sum_{l=0}^{k_\epsilon}a_l\,z^l|\le\epsilon}$ for all $|z|\le A_M$. Let $n$ be the linear dimension of $A$, and let $v=(v_1,\dots,v_{n})$ such that $\sum_j |v_j|^2=1$. One has
\begin{equation}
    0\le\langle v| B^{(k_\epsilon)}|v\rangle \le \langle v| B|v\rangle +n\,\epsilon  
\end{equation}
Therefore, $\langle v|B|v\rangle/\langle v|v\rangle \ge -n\epsilon$, for all $v$ and for all $\epsilon$, which is equivalent to saying that $B$ is positive semi-definite.
\end{proof}

\section{Additional properties of the Autoregressive Gram Hadamard Density Operator}
\subsection{An AGHDO with $R=2$ can represent any classical mixed state}\label{sec:appendix-classical}
\begin{proposition}
Let 
\begin{equation}
    \langle \sigma|\hat{\rho}_{\text{cl}}|\eta\rangle = p(\sigma)\;\delta_{\sigma,\eta}.
\end{equation}
Then, it can be exactly represented with an AGHDO with $R=2$.
\end{proposition}
\begin{proof}
We define, for $a\in\{1, 2\}$ and $\sigma_j\in\{-1,1\}$
\begin{equation}\label{eq:definition-psi-classical}
    \psi_{\sigma_{\le h},a}=\delta_{2a-3,\sigma_h}\,\sqrt{p(\sigma_h|\sigma_{<h})}
\end{equation}
where $p(\sigma_h|\sigma_{<h})$ is the conditional probability of $p(\sigma)$ defined in \cref{eq:definition-classical}. It is easy to verify that the AGHDO introduced in \cref{eq:constraint-aghdo} with $R=2$ and $\psi_{\sigma_{\le h},a}$ defined in \cref{eq:definition-psi-classical}, exactly represents the classical state of \cref{eq:definition-classical}.
\end{proof}

\subsection{An AGHDO can represent any mixed quantum state for large-enough $R$}\label{sec:appendix-universal}
We have seen that an AGHDO with $R=1$ is able to represent the ``most quantum'' state, a pure state, and by increasing $R$ to two it is also to represent an arbitrary classical state. We could be lead to think that $R=2$ should be enough to represent any other quantum state ``intermediate'' between these two cases, but a simple counting argument shows that this is not possible. Indeed, the number of degrees of freedom of an AGHDO of $N$ spins is of the order of $R\;2^N$, while a generic density operator has of the order of $4^N$ degrees of freedom. Therefore, $R$ must be, at least, of the order of $2^N$ to represent any quantum state. The following result proves that $R=2^N$ is enough to represent any quantum mixed state.

\begin{proposition}
Let $\hat{\rho}$ be a trace-one density operator of a system of $N$ spins. Then, $\hat{\rho}$ can be exactly represented by an AGHDO with $R=2^N$.
\end{proposition}
\begin{proof}
Let $\hat{\rho}$ be a density operator of a system of $n$ spins with $\text{Tr}\;\hat{\rho}=1$. From $p(\sigma):=\langle\sigma|\hat{\rho}|\sigma\rangle$, we can define the conditional probabilities $p(\sigma_h|\sigma_{< h})$. We introduce the coherence matrix $\hat{\rho}_{\text{c}}$ in the following way: for all $\sigma,\eta$ such that $p(\sigma)\neq 0$, $p(\eta)\neq 0$, we define
\begin{equation}
    \rho_{\text{c}}(\sigma,\eta) := \frac{\langle\sigma|\hat{\rho}|\eta\rangle}{\sqrt{p(\sigma)\,p(\eta)}},
\end{equation}
while, for $p(\sigma)\,p(\eta)=0$, we define $\rho_{\text{c}}(\sigma,\eta)=0$.
One has $|\rho_{\text{c}}(\sigma,\eta)|\le 1$, and $\rho_{\text{c}}(\sigma,\eta)=0$ if and only if $\rho(\sigma,\eta)=0$. The operator $\hat{\rho}_{\text{c}}$ is hermitian positive semi-definite, as for every $v=(v_1,\dots,v_{2^N})$, $v_\sigma\in\mathbb{C}$,
\begin{equation}
    \langle v|\hat{\rho}_{\text{c}}|v\rangle=\langle v'|\hat{\rho}|v'\rangle\ge 0,
\end{equation}
where we have defined $v'_\sigma=\frac{v_\sigma}{\sqrt{p(\sigma)}}$ if $p(\sigma)\neq 0$, and $v'_\sigma=0$ for $p(\sigma)=0$.

We can therefore write
\begin{equation}
    \rho_{\text{c}}(\sigma,\eta)=\sum_{a=1}^{2^N}\Psi(\sigma,a)\,[\Psi(\eta,a)]^*
\end{equation}
for some $\Psi(\sigma,a)\in\mathbb{C}$ with $\sum_{a=1}^{2^N}|\Psi(\sigma,a)|^2=1$ for all $\sigma$.
We define for $1
\le h<N$
\begin{equation}
    \psi_{\sigma_{\le h},a}:=\delta_{a,1}\;\sqrt{p(\sigma_h|\sigma_{<h})},
\end{equation}
and
\begin{equation}
    \psi_{\sigma_{\le N},a}:=\sqrt{p(\sigma_N|\sigma_{<N})}\,\Psi(\sigma,a).
\end{equation}
We also have that $\sum_{a}|\psi_{\sigma_{\le h},a}|^2=p(\sigma_{h}|\sigma_{<h})$. Using the fact that $\langle\sigma|\hat{\rho}|\eta\rangle = \sqrt{p(\sigma)\,p(\eta)}\rho_{\text{c}}(\sigma,\eta)$, one has
\begin{equation}
    \langle \sigma|\hat{\rho}|\eta\rangle =\left(\prod_{h=1}^{N-1}\psi_{\sigma_{\le h},1}\,[\psi_{\eta_{\le h},1}]^*\right) \sum_{a=1}^{2^N}\psi_{\sigma_{\le n},a}\,[\psi_{\eta_{\le n},a}]^*
\end{equation}
which is in the form of \cref{eq:definition-aghdo} with $R=2^N$.
\end{proof}

\section{Convex combination of conditional probability distributions}\label{sec:appendix-hybridization}

\begin{definition}[Convex combination of conditionals of probability distributions]
Let $p$ and $q$ be two probability distributions over a set of $N$ spins. For $0\le\alpha\le 1$, we define the convex combination of conditionals of these two probability distributions by
\begin{equation}
    \mathcal{C}_\alpha(p,q)(\sigma) := \prod_{h=1}^{N}\left(\alpha\, p(\sigma_h|\sigma_{<h})+(1-\alpha)\,q(\sigma_h|\sigma_{<h})\right)
\end{equation}
\end{definition}
\begin{proposition}
$\mathcal{C}_\alpha(p,q)$ is a probability distribution, and it can be directly sampled if $p$ and $q$ are in the autoregressive form.
\begin{proof}
$\mathcal{C}_\alpha(p,q)$ is a convex combination of $2^N$ normalized probability distributions:
\begin{equation}\label{eq:convex-combination-autoregressive}
\begin{split}
&    \mathcal{C}_\alpha(p,q)=\\
&\sum_{s_1,\dots,s_n\in\{0,1\}}\prod_{h=1}^N [\alpha\,p(\sigma_h|\sigma_{<h})]^{s_h}\,[(1-\alpha)q(\sigma_h|\sigma_{<h})]^{1-s_h},
    \end{split}
\end{equation}
therefore it is normalized. In order to directly sample from $\mathcal{C}_{\alpha}(p,q)$, we can first sample the coefficients
of the convex combination in \cref{eq:convex-combination-autoregressive} by picking $s_1,\dots,s_N$ independently with probability $p(s_h=1)=\alpha$,  and we can then sample $\sigma$ by autoregressive sampling using $p(\sigma_h|\sigma_{<h})$ if $s_h=1$ and $q(\sigma_h|\sigma_{<h})$ otherwise. 
\end{proof}
\end{proposition}

\section{Implementation details of numerical experiments} \label{append:numerical}

The architectures proposed in the article where implemented in \textsc{Jax} \cite{frostig2018jax,jax2018github} using \textsc{Flax} \cite{flax2020github}, calculations where performed using \textsc{NetKet} \cite{netket3,netket2,Vicentini2021Nat} and distributed across different nodes with \textsc{mpi4jax} \cite{Vicentini2011mpi4jax}.

The Autoregressive GHDO of \cref{eq:def:gram-hadamard-ndo} has been realized by implementing $\psi_{\sigma_{\le h}, a}$ by alternating Masked-Dense layers with variable feature density 
and \textrm{SELU} nonlinearities~\cite{Kaluber2017SELU}. 
Dense layers were taken with complex parameters initialized according to a truncated normal distribution with width between $10^{-3}$ and $10^{-1}$, while selu acts independently on the real and imaginary part of its inputs.

The linear system of the TDVP Equation has been solved using conjugate gradients method and a large regularisation between $[10^{-4},10^{-2}]$.

\bibliographystyle{apsrev4-1}
\bibliography{biblio}

\end{document}